\documentclass{lawcg}

\usepackage[latin1]{inputenc}
\usepackage{amsmath,amssymb}
\usepackage{url}

\usepackage{indentfirst}

\usepackage[vlined,linesnumbered,titlenumbered,algo2e,ruled]{algorithm2e}

\newtheorem{theorem}{Theorem}
\newtheorem{proposition}{Proposition}

\newenvironment{proof}{\par \noindent \textit{Proof (Sketch)}.}{\hfill $\Box$\newline}

\title{Polynomial enumeration of chordless cycles on cyclically orientable graphs}

 \author{Diane Castonguay  \and
       \speaker{Elis\^{a}ngela Silva Dias}\thanks{Partially supported by FAPEG -- Funda\c{c}\~{a}o de Amparo \`{a} Pesquisa do Estado de Goi\'{a}s}}

 \institute{{\it \{diane,elisangela\}@inf.ufg.br} \\Instituto de Inform\'{a}tica -- Universidade Federal de Goi\'{a}s -- UFG}

\begin{document}

\maketitle

In a finite undirected simple graph, a {\it chordless cycle} is an induced subgraph which is a cycle. A graph is called cyclically orientable if it admits an orientation in which every chordless cycle is cyclically oriented. We propose an algorithm to enumerate all chordless cycles of such a graph. Compared to other similar algorithms, the proposed algorithm have the advantage of finding each chordless cycle only once in time complexity $\mathcal{O}(n^2)$ in the input size, where $n$ is the number of vertices.

\section{Introduction}
\label{sec:Introduction}

Given a finite undirected simple graph $G$, a {\it chordless cycle} is an induced subgraph that is a cycle. 
A chordless cycle with four of more edges is called {\it hole}.

A solution to the problem of determining if a graph contains a chordless cycle of length $k \geq 4$, for some fixed value of $k$, was proposed by Hayward~\cite{H1987}. Golumbic~\cite{G1980} proposed an algorithm to recognize chordal graphs, that is, graphs without any chordless cycles. The case for $k \geq 5$ was settled by Nikolopoulos and Palios~\cite{NP2007}. 

It is important observe that to find a unique chordless cycle of length $k$ is easier than to enumerate all chordless cycles in a graph $G$. However, enumeration is a fundamental task in computer science and many algorithms have been proposed for enumerating graph structures such as 
cycles~\cite{RT1975,W2008}, circuits~\cite{B2010,T1973}, paths~\cite{HH2006,RT1975}, trees~\cite{KR2000,RT1975} and cliques~\cite{MU2004,TTT2006}. Due to the number of cycles -- which can be exponentially large -- these kind of tasks are usually hard to deal with, since even a small graph may contain a huge number of such structures. 




An algorithm to enumerate chordless cycles, with $\mathcal{O}(n + m)$ time complexity in the output size, was proposed by Uno and Satoh~\cite{U2014} and each chordless cycle will appears more than once in the output. Actually, each cycle will appear as many times as its length. Thus, the algorithm has $\mathcal{O}(n \cdot (n + m))$ time complexity in size of the sum of lengths of all the chordless cycles in the graph.

Dias et al.~\cite{DCLJ2014} proposed two algorithms to enumerate all chordless cycles of a given graph $G$, with $\mathcal{O}(n+m)$ time complexity in the output size, with the advantage of finding each chordless cycle only once. The core idea of algorithms is to use a vertex labeling scheme, with which any arbitrary cycle can be described in a unique way. With this, they generate an initial set of vertex triplets and use a DFS strategy to find all the chordless cycles. 

Cyclically orientable (CO) graphs are introduced by Barot et al. in~\cite{BGZ2006}. A graph $G$ is CO if it admits an orientation in which any chordless cycle is cyclically oriented. Such an orientation is also called cyclic. The authors obtained several nice characterizations of CO-graphs, being motivated primarily by their applications in cluster algebras. Gurvich~\cite{G2008} and Speyer~\cite{S2005} obtained several new characterizations that provide algorithms for recognizing CO-graphs and obtaining their cyclic orientations in linear time. For CO-graphs, we show that the amount of chordless cycles is polynomial in the input size.

We present an algorithm that verifies whether the given graph is cyclically orientable and, in positive case, that enumerates all chordless cycles in polynomial time.

The remainder of the paper is organized as follows: some preliminaries definitions and comments are presented in Section~\ref{sec:Preliminaries}; our algorithm are introduced in Section~\ref{sec:Algorithm}; Section~\ref{sec:Correctness} shows the correctness of the algorithm and analyze time and space complexity; finally, in Section~\ref{sec:Conclusions} we draw our conclusions.

\section{Preliminaries}
\label{sec:Preliminaries}

In this section, we present the mathematical definitions that support our approach to enumerate all chordless cycles of a cyclically orientable graph.

Let $G = (V(G), E(G))$ be a finite undirected simple graph with vertex set $V(G)$ and edge set $E(G)$. Let $n=|V(G)|$ and $m=|E(G)|$. We denote by $Adj(x) = \{y \in V(G) | (x, y) \in E(G)\}$.

A \textit{simple path} is a finite sequence of vertices $\langle v_1, v_2, \dots, v_k \rangle$ such that $(v_i, v_{i+1})\allowbreak \in E(G)$ and no vertex appears repeated in the sequence, that is,  $v_i\neq v_j$, for $i = 1, \dots, k-1$, $j = 1, \dots, k$ and $j\neq i$. A \textit{cycle} is a simple path $\langle v_1, v_2, \dots, v_k \rangle$ such that $(v_k, v_1)\in E(G)$. Note that our definition of cycle, as in~\cite{DCLJ2014}, does not repeat the first vertex at the end of the sequence as usually done by other authors. A {\it chord} of a path (resp. cycle) is an edge between two vertices of the path (cycle), that is not part of the path (cycle). A path (cycle) without chord is called a {\it chordless path (chordless cycle)}.


A graph $G$ is called \textit{connected} when there exists a path between each pair of vertices of $G$, otherwise $G$ is called {\it disconnected}. A {\it connected component} of a graph $G$ is a maximal connected subgraph of $G$. A graph is called {\it two-connected} if it is connected and is necessary the elimination of at least two vertices to disconnected it.

Two-connected components are important because any chordless cycle is contained in exactly one. To calculate them, we can use an algorithm based in Szwarcfiter~\cite{S1988} ideas, that has time complexity $\mathcal{O}(n^2)$.

For better understanding of this work, we will present a theorem and a proposition that is used in our algorithm.


\begin{theorem} [Speyer \cite{S2005}] \label{teo1_Speyer}
A graph $G$ is cyclically orientable if and only if all of its two-connected
components are. A two-connected graph is cyclically orientable if and only if it is either a cycle, a single edge, or of the form $G' \cup C$ where $G'$ is a cyclically orientable graph, $C$ is a cycle and $G'$ and C meet along a single edge. Moreover, if $G=G' \cup C$ is any such decomposition of $G$ into a cycle and a subgraph meeting along a single edge, then $G$ is cyclically orientable if and only if $G'$ is.
\end{theorem}

\begin{proposition} [Speyer~\cite{S2005}] \label{prop:LimiteDeArestas}
Let $G$ be a cyclically orientable graph with $n$ vertices. Then $G$ has at
most $2 \cdot n - 3$ edges.
\end{proposition}



\section{The proposed algorithm}
\label{sec:Algorithm}

Based in theorems and propositions described by Speyer~\cite{S2005}, the Algorithm~\ref{alg:CiclicamenteOrientado} is able to verify if a given graph $G$ is cyclically orientable and, in positive case, to return all chordless cycles as we show in Theorem~\ref{teo:CO_CC}.

The work of the algorithm is based in the analysis of each two-connected component found to a given graph as input. Following exactly the idea of Theorem~\ref{teo1_Speyer}, the algorithm identifies, for a two-connected component, chordless cycles. This is done aiming to reduce the initial two-connected components to a unique cycle.

The Algorithm~\ref{alg:CiclicamenteOrientado}, initially, verifies if the given graph meets the Proposition~\ref{prop:LimiteDeArestas}, that is, if the graph has $2 \cdot n - 3$ edges. If not, it returns NO. Next, it finds all two-connected components and it also verifies if each component meets the Proposition~\ref{prop:LimiteDeArestas} or if the vertices of graph not have vertices with degree two. If one this conditions are not satisfied, it returns NO.

After to do the preliminaries verifications, the algorithm storages in a queue $F$ all vertices of degree two to each two-connected component. Vertices are removed and new are added to $F$ as soon as the algorithm runs. To add and to remove elements of $F$ takes time $\mathcal{O}(1)$. This continues to occur until all vertices of degree equal two are visited. Observe that if $G$ is CO then all vertices will pass exactly once in $F$.

The algorithm tries, starting in vertices of queue $F$, find and eliminate paths (cycles) up to reduce the initial two-connected component to a cycle and, then, to decide if it is CO. After verify that a two-connected component is CO, the algorithm will analyse the next component. This will continue for all components. In final of process, the given graph will be classified as CO if all two-connected components receive the CO classification; otherwise, the graph is classified as not CO.

The algorithm returns YES if and only if all two-components returns YES. Therefore, given a two-connected graph $G$, it determines, in $\mathcal{O}(n^2)$ complexity time, whether this graph is CO and, then, it returns the set of all chordless cycles $C$ of $G$.


\begin{center}
\begin{minipage}{ \textwidth}
\begin{algorithm2e}[H]
\SetKwFunction{comp}{Two-ConnectedComponents}
 \BlankLine
 \KwIn{An undirected simple graph $G$.}
 \KwOut{Response if $G$ is CO and, if it is, the set $C$ of chordless cycles.}
 \BlankLine

 \eIf{$(|E(G)| > 2 \cdot |V(G)| - 3)$}{
 	\Return NO
 }
 {
	
	\ForEach{$\mbox{two-connected component } G_i \mbox{ of } G$}{
		\If{$(|E(G_i)| > 2 \cdot |V(G_i)| - 3)$}{
 			\Return NO
		}
	}
	
	$C \leftarrow \varnothing$\\

	\ForEach{$\mbox{two-connected component } G_i \mbox{ of } G \mbox{ that is not a single edge}$}{

 		{\bf initialize} the queue $F$ with all vertices of $degree(v) = 2$\\	

		\While{$(F \mbox{ is not empty)}$}{	
			{\bf take} the first element $u$ of queue $F$\\
		
			\If{$(color(u) = white)$}{
				$P \leftarrow \varnothing; \quad y \leftarrow u$\\
				$x \leftarrow a$, such that $a \in Adj(u)$ and $color(a) = white$\\

				
				\While{$((degree(x) = 2) \mbox{ and } (\exists  a \in Adj(x): color(a) = white))$}{
					$F \leftarrow F - \{x\}; \quad color(x) \leftarrow gray$\\
					$P \leftarrow \langle key(x), P \rangle; \quad x \leftarrow a$\\
				}
				
				\While{$((degree(y) = 2) \mbox{ and } (\exists  b \in Adj(y): color(b) = white))$}{
					$F \leftarrow F - \{y\}; \quad color(y) \leftarrow gray$\\
					$P \leftarrow \langle P, key(y) \rangle; \quad y \leftarrow b$\\
				}
				
				\eIf{$(x \neq y)$}{ 
					\eIf{$((x, y) \in E(G_i))$}{
						$C \leftarrow C \cup \langle key(x), P, key(y) \rangle$\\
						$degree(x) \leftarrow degree(x) - 1; \quad degree(y) \leftarrow degree(y) - 1$\\
						\If{$(degree(x) = 2)$}{
							$F \leftarrow F \cup \{x\}$\\
						}
						\If{$(degree(y) = 2)$}{
							$F \leftarrow F \cup \{y\}$\\
						}
					}
					{\tcp{we create a new vertex $w$.}													$Adj(x) \leftarrow Adj(x) \cup \{w\}; \quad Adj(y) \leftarrow Adj(y) \cup \{w\}$\\
							$Adj(w) \leftarrow \{x, y\}; \quad degree(w) \leftarrow 2$\\
							$color(w) \leftarrow white; \quad key(w) \leftarrow P$\\
					}
				}
				{
					$C \leftarrow C \cup \langle key(x), P \rangle$\\
				}
			}	
		}
		\ForEach{$u \in V(G_i)$}{
			\If{$color(u) = white$}{
				\Return NO
			}
		}
	}
}
\Return YES, $C$

\BlankLine

\caption{$ChordlessCyclesCOGraph(G)$ \label{alg:CiclicamenteOrientado}}
\end{algorithm2e}
\end{minipage}
\end{center}

\section{Algorithm analysis}
\label{sec:Correctness}

The correctness of Algorithm $ChordlessCyclesCOGraph(G)$ is divide in two parts. The first one of recognizing if $G$ is CO follows from Speyer~\cite{S2005}. The theorem below complete the correctness of algorithm.

\begin{theorem} \label{teo:CO_CC}
If a graph $G$ is CO, then the Algorithm~\ref{alg:CiclicamenteOrientado} finds all chordless cycles of $G$.
\end{theorem}

\begin{proof}
Suppose to $G$ is CO. Since all two-connected components $G_i$ of $G$ are CO, we can assume that $G$ is two-connected. Denote by $G'$ the graph obtained at the end of an iteration. In the first case (line 21), we have that $G =  G' \cup C$. All chordless cycles of $G$ are chordless cycles of $G'$ or equal to $C$, since other cycles that contain vertices of the path $P$ will have a chord $(x, y)$. In the second case (line 29), we have that $G'$ is essentially $G$, since we identify the new vertex $w$ with $P$. In the last case (line 33), the graph $G$ is a cycle which is clearly a chordless cycle.
\end{proof}


The algorithm to determinate all two-connected components has time complexity $\mathcal{O}(m)$, see~\cite{S1988}. Based in Proposition~\ref{prop:LimiteDeArestas}, the algorithm starts testing if $G$ has at most $2 \cdot n - 3$ edges. Therefore, any computation takes time $\mathcal{O}(m)$ and, in fact, has time $\mathcal{O}(n)$. 

Our algorithm uses a boolean function $color(v)$ which assigns the value ``white'' or ``gray'' to all vertices. The ``gray'' vertices are those that we remove from $G$ and will be identify with the new vertex $w$ or will compose a new chordless cycle. If $G$ is CO, then all vertices will enter at some stage in $F$ and afterwards will be colored with ``gray''. The algorithm has $\mathcal{O}(n)$ steps and it resolves recursively the same problem, using DFS. The DFS algorithm has time complexity $\mathcal{O}(n+m)$. Therefore, the Algorithm~\ref{alg:CiclicamenteOrientado} has time complexity $\mathcal{O}(n^2)$.

\section{Conclusions}
\label{sec:Conclusions}

We presented an algorithm easy to follow to enumerate all chordless cycles in a CO-graph, that has time complexity $\mathcal{O}(n^2)$. The core idea is to go diminishing the given graph, listing the chordless cycles in the process, until we have just a single cycle and, then, we know that the graph is CO.

%

\end{document}